\documentclass[12pt]{article}

\usepackage{verbatim}
\usepackage{graphicx}
\usepackage{amsfonts}
\usepackage{amsmath,amsthm}
\usepackage{epsfig}
\usepackage{psfrag}
\usepackage{graphics}
\usepackage{subfigure}
\usepackage{lscape}

\usepackage{caption}
\usepackage{everysel}
\usepackage{keyval}
\usepackage{ragged2e}

\newcommand\independent{\protect\mathpalette{\protect\independent}{\perp}}
\def\independent#1#2{\mathrel{\rlap{$#1#2$}\mkern2mu{#1#2}}}

\newcommand{\mo}{\,\,\mathrm{mod}\,\,}
\newcommand{\wtil}{\widetilde}

\newcommand{\F}{\mathbb{F}}

\newcommand{\mZ}{\mathbb{Z}}

\newcommand{\pp}{\mathbb{P}}

\newtheorem{definition}{Definition}
\newtheorem{thm}{Theorem}
\newtheorem{lemma}{Lemma}

\usepackage{times}

\newenvironment{sciabstract}{%
\begin{quote} \bf}
{\end{quote}}

\newcounter{lastnote}

\title{~~\\ [-1.5in]
Privacy-Preserving Methods for Sharing \\ [.125in]
Financial Risk Exposures\footnote{We thank Arnout Eikeboom, Martin Hirt,
Bob Merton, and Ron Rivest for helpful comments and discussion.  The
views and opinions expressed in this article are those of the authors only,
and do not necessarily represent the views and opinions of AlphaSimplex,
EPFL, MIT, any of their affiliates or
employees, or any of the individuals acknowledged above.
{\bf Disclosures}:  In addition to his academic affiliation, A.~Lo is also Chief Investment
Strategist of AlphaSimplex Group, LLC, a consultant to the Office of
Financial Research, a research associate of the National Bureau of
Economic Research, a member of FINRA's Economic Advisory Committee, the
New York Fed's Financial Advisory Roundtable, Moody's Academic Advisory
and Research Committee, and Beth Israel Deaconness Medical Center's Board
of Overseers.  In addition to his academic affiliation, A.~Khandani is also an associate
at Morgan Stanley.}}

\author
{Emmanuel A.\ Abbe,$^{1,2}$ Amir E.\ Khandani,$^{3}$ Andrew W. Lo$^{2,3,4\ast}$\\
\\
\normalsize{$^{1}$EPFL School of Communication and Computer Sciences}\\ [-.07in]
\normalsize{INR130 Station 14, Lausanne 1015, Switzerland}\\
\normalsize{$^{2}$MIT CSAIL and EECS}\\ [-.07in]
\normalsize{32 Vassar Street, Cambridge, MA 02139, USA}\\
\normalsize{$^{3}$MIT Laboratory for Financial Engineering}\\ [-.07in]
\normalsize{100 Main Street, Cambridge, MA 02142, USA}\\
\normalsize{$^{4}$AlphaSimplex Group, LLC}\\ [-.07in]
\normalsize{One Cambridge Center, Cambridge, MA 02142, USA}\\
\normalsize{$^\ast$Corresponding author.  Email:  alo@mit.edu}
}

\date{November 19, 2011}

\begin{document}


\baselineskip24pt


\maketitle

~~\vspace{-.8in}

\thispagestyle{empty}

\begin{sciabstract}
\baselineskip 14pt%
Unlike other industries in which intellectual property is patentable, the
financial industry relies on trade secrecy to protect its business
processes and methods, which can obscure critical financial risk exposures
from regulators and the public.  We develop methods for sharing and
aggregating such risk exposures that protect the privacy of all parties
involved and without the need for a trusted third party.  Our approach
employs secure multi-party computation techniques from cryptography in
which multiple parties are able to compute joint functions without
revealing their individual inputs.  In our framework, individual financial
institutions evaluate a protocol on their proprietary data which cannot be
inverted, leading to secure computations of real-valued statistics such a
concentration indexes, pairwise correlations, and other single- and
multi-point statistics. The proposed protocols are computationally
tractable on realistic sample sizes. Potential financial applications
include: the construction of privacy-preserving real-time indexes of bank
capital and leverage ratios; the monitoring of delegated portfolio
investments; financial audits; and the publication of new indexes of
proprietary trading strategies.
\end{sciabstract}
\newpage

\setcounter{page}{1}

\section*{Introduction}
While there is still considerable controversy over the root causes of the
Financial Crisis of 2007--2009, there is little dispute that regulators,
policymakers, and the financial industry did not have ready access to
information with which early warning signals could have been generated.
For example, prior to the Dodd Frank Act of 2010, even systemically
important financial institutions such as AIG and Lehman Brothers were not
obligated to report their amount of financial leverage, asset illiquidity,
counterparty risk exposures, market share, and other critical risk data to
any regulatory agency.  If aggregated over the entire financial industry,
such data could have played a critical role in providing regulators and
investors with advance notice of AIG's unusually concentrated position in
credit default swaps, as well as the exposure of money market funds to
Lehman bonds. Of course, such information is currently considered
proprietary and highly confidential, and releasing it into the public
domain would clearly disadvantage certain companies and benefit their
competitors.  But without this information, regulators and investors
cannot react in a timely and measured fashion to growing threats to
financial stability, thereby assuring their realization.

At the heart of this vexing challenge is privacy. Unlike other industries
in which intellectual property is protected by patents, the financial
industry consists primarily of ``business processes'' that the U.S.\
Patent Office deems unpatentable, at least until recently
\cite{lerner:2002}.  Therefore, trade secrecy has become the preferred
method by which financial institutions protect the vast majority of their
intellectual property, hence their need to limit disclosure of their
business processes, methods, and data. Forcing a financial institution to
publicly disclose its proprietary information---and without the quid pro
quo of 17-year exclusivity that a patent affords---will obviously
discourage innovation, which benefits no one.  Accordingly, government
policy has tread carefully on the financial industry's disclosure
requirements.

In this paper, we propose a new approach to financial systemic risk
management and monitoring via cryptographic computational methods in which
the two seemingly irreconcilable objectives of protecting trade secrets
and providing the public with systemic risk transparency can be achieved
simultaneously.
To accomplish these goals, we develop self-regulated protocols for
securely computing aggregate risk measures. The protocols are constructed
using secure multi-party computation tools
\cite{gmw,ccd,cddhr,bmr,yao,bgw}, specifically using secret sharing
\cite{shamir}. It is known from \cite{yao,gmw} that general Boolean
functions can be securely computed using ``circuit evaluation protocols''.
Since computing any function on real-valued data is approximated
arbitrarily well by computing a function on quantized (or binary) data,
such an approach can theoretically be used. However, for arbitrary
functions and high precision, the resulting protocols may be
computationally too demanding and therefore impractical. We show in this
paper that for computing aggregate risk measures based on standard sample
moments such as means, variances, and covariances---the typical inputs for
financial risk measures---simple and efficient protocols can be achieved
using secret-sharing over large finite fields or directly over the reals.

With the resulting measures, it is possible to compute the aggregate risk
exposures of a group of financial institutions---for example, a
concentration (or ``Herfindahl'') index of the credit default swaps
market, the aggregate leverage of the hedge-fund industry, or the
margin-to-equity ratio of all futures brokers---without jeopardizing the
privacy of any individual institution.  More importantly, these measures
will enable regulators and the public to accurately measure and monitor
the amount of risk in the financial system while preserving the
intellectual property and privacy of individual financial institutions.

Privacy-preserving risk measures may also facilitate the ability of the
financial industry to regulate itself more effectively.  Despite the long
history of ``self-regulatory organizations'' (SROs) in financial services,
the efficacy of self regulation has been sorely tested by the recent
financial crisis. However, SROs may be considerably more effective if they
had access to timely and accurate information about systemic risk that did
not place any single stakeholder at a competitive disadvantage. Also, the
broad dissemination of privacy-preserving systemic risk measures will
enable the public to respond appropriately as well, reducing general
risk-taking activity as the threat of losses looms larger due to
increasing systemic exposures. Truly sustainable financial stability is
more likely to be achieved by such self-correcting feedback loops than by
any set of regulatory measures.

\section*{Secure Protocols}
Many important statistical measures such as, mean, standard deviation,
concentration ratios, pairwise correlations can be obtained by taking summations and inner
products on the data. Therefore, we present secure protocols for these two
specific functions.

We start with a basic protocol to securely compute the sum of $m$ secret
numbers. This protocol result from an application of secret-sharing \cite{shamir} and 
basic probability results. 
We assume that each number belongs to a known range, which we
pick to be $[0,1]$ for simplicity. Recall that the operation $a$ modulo
$m$ (written $a \mo m$) produces the unique number $a+ km \in [0,m)$ where
$k$ is an integer, e.g., $3.6 \,\,\mathrm{mod}\, 2=1.6$.

\subsection*{Secure-Sum Protocol}
For $i=1,\dots,m$, each party $i$ possesses the secret number $x_i \in
[0,1]$ as an input, and the output to
each party is $s=\sum_{i=1}^m x_i$ (where the addition is over the reals).\\
The protocol is as follows:
\begin{enumerate}
\item Each pair of parties exchange privately random numbers. Namely,
    for all $i,j$ with $i\neq j$, party $i$ provides to party $j$ a
    random number $R_{ij}$ drawn uniformly at random in $[0,m]$.
\item For each $i$, party $i$ adds to its secret number the random
    numbers it has received from other parties and subtracts the
    random numbers it has provided to other parties.
More formally, party $i$ computes $S_i=x_i+\sum_{j\in \{1,\dots,m\}
\atop{j \neq i}} R_{ji}  - \sum_{j\in \{1,\dots,m\} \atop{j \neq i}}
R_{ij} \,\,\mathrm{mod}\, m$. Each party publicly reveals $S_i$.
\item Each party computes $S=\sum_{i=1}^m S_i \,\,\mathrm{mod}\,m$, which equals $s=\sum_{i=1}^m x_i$.
\end{enumerate}

\vspace{2ex}\noindent%
{\bf Numerical example.} Let $m\!=\!3$ (i.e., three parties),
$x_1\!=\!0.1$, $x_2\!=\!0.2$ and $x_3\!=\!0.3$. In the first round of the
protocol, the parties exchange random numbers $R_{ij}$. For example,

\begin{center}
  \begin{tabular}{ l || c | c | c | }
    &Party 1 & Party 2& Party 3 \\ \hline \hline
   Party 1 provides  && 1.4 & 2.1 \\ \hline
  Party 2 provides&1.1&&2.3 \\ \hline
    Party 3 provides& 0.3& 2.9& \\ \hline
      \end{tabular}
\end{center}

\noindent%
In the second round, party $i$ adds to its secret number the
elements of the $i$-th column and subtract the elements of the $i$-th row
(using modulo 3 arithmetic). Each party publishes the result $S_i$:

\begin{center}
  \begin{tabular}{ | c | c | c | }
    $S_1$ & $S_2$& $S_3$ \\ \hline \hline
  1 & 1.1 & 1.5 \\ \hline
      \end{tabular}
\end{center}
Finally, the parties add these numbers (modulo 3) and compute the output
sum:
$$s=3.6 \mo 3 = 0.6.$$

\vspace{2ex}\noindent%
{\bf Protocol correctness and secrecy}. If the parties follow the protocol
correctly, it is easy to check that the correct sum is always obtained,
since each element $R_{ij}$ is added and subtracted once in $S$. In
addition, we show that this protocol reveals nothing else about the secret
numbers than their sum, even if the parties attempt to infer more from the
exchanged data. For example, Party 1 may try to learn more about other
parties' secret numbers by using the information gathered in
$S_1,S_2,S_3$. We state informally the secrecy guarantee in the following
theorem and provide a formal statement and proof in the appendix. We first
illustrate a weaker fact here by plotting the values of $S_1,S_2,S_3$ for
several realizations of the random numbers $R_{ij}$, while keeping fixed
$x_1\!=\!0.1$, $x_2\!=\!0.2$ and $x_3\!=\!0.3$. As shown in Figure
\ref{slices}, the realizations of $(S_1,S_2,S_3)$ uniformly cover the set
of points $(s_1,s_2,s_3)$ for which $s_1\!+\!s_2\!+\!s_3
\,\,\mathrm{mod}\, 3=0.6$, suggesting that there is no relevant
information in the $S_i$'s other than their sum.

\begin{figure}
\begin{center}
\includegraphics[scale=.7]{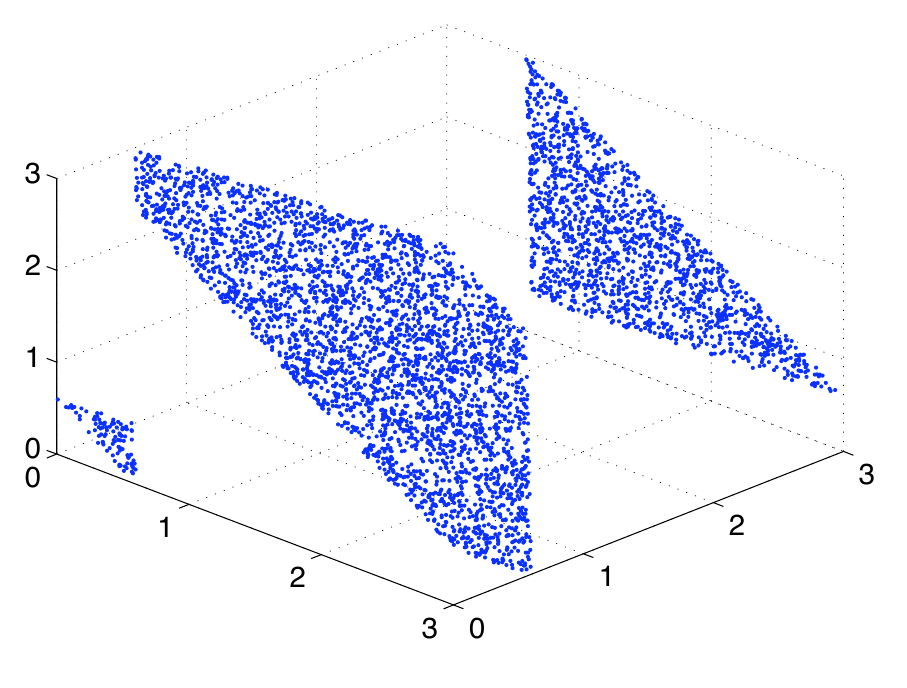}
\caption{Each point in the plot is a realization of $(S_1,S_2,S_3)$ (step 2 in the
Secure-Sum protocol) for a drawing of the matrix $R$, keeping $x_1=0.1$,
$x_2\!=\!0.2$ and $x_3\!=\!0.3$ fixed. As illustrated by the plot, the set of
points $(s_1,s_2,s_3)$ for which $s_1\!+\!s_2\!+\!s_3 \,\,\mathrm{mod}\,3=0.6$ is
uniformly covered, suggesting that the $S_i$'s do not carry any other
information about the $x_i$'s than their sum.} \label{slices}
\end{center}
\end{figure}

The following is obtained assuming that parties follow the protocol requirements without deviating from it. 
\begin{thm}\label{thm-sum-inform}
The Secure-Sum protocol outputs the sum of $m$ privately owned real
numbers and does not reveal any additional information about the
individual numbers.
\end{thm}
This theorem follows directly from secret-sharing \cite{shamir} and basic probability results.
For convenience, we provide a proof in the Appendix.

\subsection*{Secure-Inner-Product Protocol}
To compute securely the inner product of two real vectors, slightly more
sophisticated protocols are developed and presented in the appendix, using
basic secret sharing \cite{shamir}, secret-sharing as employed in
\cite{bgw,ccd,cddhr}, and Oblivious Transfer \cite{rabinOT,OT}.
The variants include information-theoretic and cryptographic protocols on
quantized or real data, and have different attributes discussed in the
appendix. We state here an informal result regarding one of these
protocols which we call Secure-Inner-Product protocol 1.

\begin{thm}\label{thm-ip-inform}
The Secure-Inner-Product protocol 1 outputs the sum of two privately owned
quantized vectors and does not reveal any additional information about the
individual vectors.
\end{thm}

\noindent%
Note that the previous two theorems hold provided that the parties follow
the protocol requirements (without colluding or cheating). Extensions to malicious parties or other type
of functions can also be developed but are not discussed here.

\section*{Illustrative Example}
%


To illustrate the practical implementation of privacy-preserving measures,
we provide a simple numerical example using publicly available quarterly
data from June 1986 to December 2010 (released in arrears by the U.S.\
Federal Reserve) on the total amount of outstanding loans linked to real
estate issued by three major bank holding companies: Bank of America,
JPMorgan, and Wells Fargo \cite{bhc_data}. Suppose that the aggregate
value of these loans across the three banks is the risk exposure of
interest, and the magnitude of outstanding loans for each bank is the
proprietary data to be kept private.  The historical time series of these
data are displayed in Figure \ref{fig_bhc:subfig1}; the bar graph in blue
is the aggregate risk exposure to be computed and the three line graphs
are the proprietary inputs.

The desired result can be obtained with an application of the Secure-Sum
protocol described above \cite{protocol_note}, which consists of two
steps. In the first step, each institution produces two random numbers to
be shared, one for each of the other two participating institutions. These
numbers are shown in line graphs of Figure \ref{fig_bhc:subfig2} where the
color coding indicates the institution generating the random numbers.
Since these numbers are purely random, there is no relationship between
them and the private data of Figure \ref{fig_bhc:subfig1}, a fact that is
clear from visual inspection of the intermediate outputs in Figure
\ref{fig_bhc:subfig2}.

In the second step of the Secure-Sum protocol, each institution uses its
private data, the two numbers it receives from the other two participating
banks, as well as the two numbers it sends to the other two institutions
to produce a single value, which we refer to as the privacy-preserving
measure of its private data.  This value will be revealed to the other two
institutions. While these privacy-preserving measures, shown in Figure
\ref{fig_bhc:subfig3}, seem like a pure noise, they have just enough of
the original data so that the sum of these three numbers under modulo
arithmetic yields the correct sum of the original inputs. The key here is
that the randomness produced in the first step, as shown in Figure
\ref{fig_bhc:subfig2}, {\it exactly\/} cancels in the second step due to
the way that the protocol in constructed. It is apparent that the
aggregate loans outstanding in Figure \ref{fig_bhc:subfig3} is identical
to the corresponding graph in Figure \ref{fig_bhc:subfig1}, but the former
graph has been computed using only the privacy-preserving measures of
Figure \ref{fig_bhc:subfig3}.

Despite the fact that the underlying data used in this example is not
confidential, even in this simple illustrative case privacy-preserving
measures may still prove useful in providing financial institutions and
regulators with an incentive to release the data without a lag.  More
timely releases would obviously benefit all stakeholders by allowing them
to respond more nimbly to changing market conditions, but such releases
could also disadvantage certain parties in favor of others if privacy were
not assured.  Moreover, this example underscores the simplicity with which
more sensitive data such as leverage ratios, positions in illiquid assets,
and off-balance-sheet derivatives holdings can be shared regularly,
securely, and in a timely fashion.

We consider only three institutions in this example because it is the
simplest non-trivial case in which privacy-preserving measures of
aggregate sums can be constructed.  Clearly, the protocol is applicable
for any number of participants greater than two, and implementation for
even several thousand participants is extremely fast. More complex risk
exposures such as Herfindahl concentration indices require two
applications of the Secure-Sum protocol, but the computational burdens are
still quite modest. The Secure-Inner-Product protocol can be used to
construct multi-point statistical measures such as average correlations
between changes in securities holdings or leverage across industry
participants.

\begin{figure}[htbp]
\centering
\begin{minipage}{\linewidth}
\begin{minipage}{.45\linewidth}
\subfigure[]{
\includegraphics[scale=0.3,clip,trim=0 90 0 0]{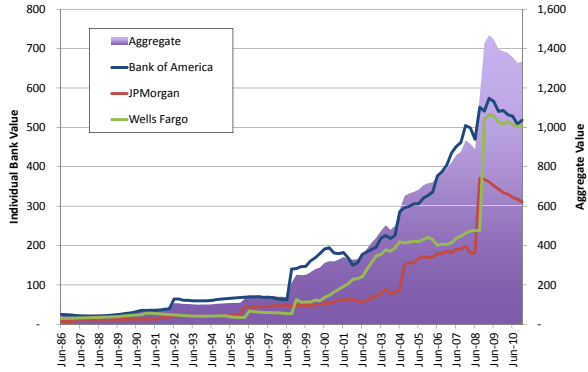}
\label{fig_bhc:subfig1}}
\end{minipage}
~~\hspace{.5in}
\begin{minipage}{.45\linewidth}
\subfigure[]{
\includegraphics[scale=0.3,clip,trim=0 75 0 0]{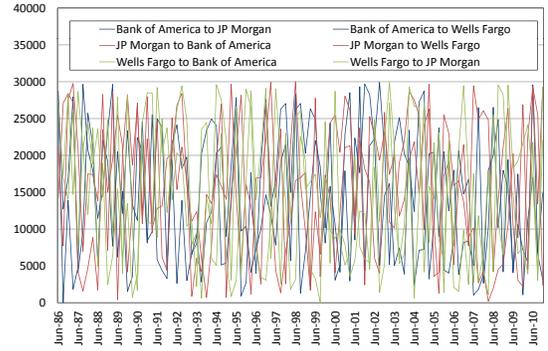}
\label{fig_bhc:subfig2}}
\end{minipage}
\\ [-.25in]
\begin{minipage}{.45\linewidth}
\subfigure[]{
\includegraphics[scale=0.33,clip,trim=35 75 0 0]{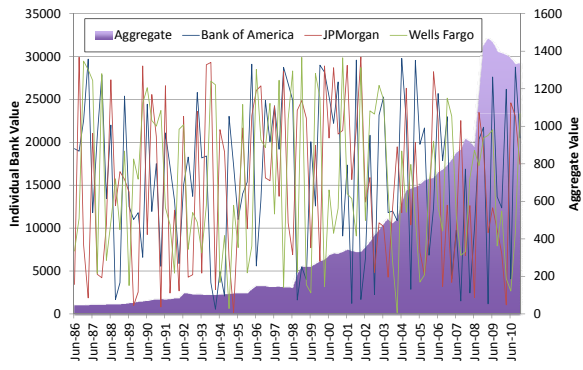}
\label{fig_bhc:subfig3}}
\end{minipage}
\end{minipage}
\caption{An illustration of a privacy-preserving measure of
the aggregate amount of real-estate-linked loans outstanding for
Bank of America, JPMorgan, and Wells Fargo from June 1986 to December 2010.
Panel (a) contains the three historical quarterly time series of outstanding
outstanding loans which is private and the aggregate sum which we wish to
compute securely.  Panel (b) contains the six time series of intermediate
numbers that are exchanged bilaterally between all pairs.  Panel (c)
contains the three privacy-preserving values that are shared between all banks
and used to compute the aggregate sum, which is identical to the aggregate
sum in Panel (a).} \label{fig_bhc}
\end{figure}

\section*{Discussion}
By construction, privacy-preserving measures of financial risk exposures
cannot be ``reverse-engineered'' to yield information about the individual
constituents.  Accordingly, there is no guarantee that the individual
inputs are truthful.  In this respect, the potential for misreporting and
fraud are no different for these measures than they are for current
reporting obligations by financial institutions to their regulators, and
existing mechanisms for ensuring compliance---random periodic examinations
and severe criminal and civil penalties for misleading disclosures---must
be applied here as well.

However, unlike traditional regulatory disclosures, privacy-preserving
measures will provide its users with a strong incentive to be truthful
because the mathematical guarantee of privacy eliminates the primary
motivation for obfuscation.  Since each institution's proprietary
information remains private even after disclosure, dishonesty yields no
discernible benefits but could have tremendous reputational costs, and
this asymmetric payoff provides significantly greater economic incentive
for compliance.  Moreover, accurate and timely measures of system-wide
risk exposures can benefit the entire industry in allowing institutions
and investors to engage in self-correcting behavior that can reduce the
likelihood of systemic shocks.  For example, if all stakeholders were able
to monitor the aggregate amount of leverage in the financial system at all
times, there is a greater chance that market participants would become
more wary and less aggressive as they observe leverage rising beyond
prudent levels.

A related issue is whether participation in privacy-preserving disclosures
of financial risk exposures is voluntary or mandated by regulation.  Given
the extremely low cost/benefit ratio of such disclosures, there is reason
to believe that the financial industry may well adopt such disclosures
voluntarily.  A case in point is Markit, a successful industry consortium
of dealers of credit default swaps (CDS) that emerged in 2001 to pool
confidential pricing data on individual CDS transactions and make the
anonymized data available to each other and the public so as to promote
transparency and liquidity in this market \cite{markit}.  According to
Markit's website, the data of its consortium members are
``$\ldots$provided on equal terms to whoever wanted to use it, with the
same data released to all customers at the same time, giving both the
sell-side and buy-side access to exactly the same daily valuation and risk
management information''.  From this carefully crafted statement, it is
clear that equitable and easy access to data is of paramount importance in
structuring this popular data-sharing consortium.  Privacy-preserving
methods of sharing information could greatly enhance the efficacy and
popularity of such cooperatives.

The same motivation applies to the sharing of aggregate financial risk
exposures, but with even greater stakes as the recent financial crisis has
demonstrated.  Once a privacy-preserving system-risk-exposures consortium
is established, the benefits will so clearly dominate the nominal costs of
participation that it should gain widespread acceptance and adoption in
short order.  Indeed, participation in such a consortium may serve as a
visible commitment to industry best practices that yields tangible
benefits for business development, leading to a ``virtuous cycle'' of
privacy-preserving risk disclosure throughout the financial industry

\section*{Conclusion}

Privacy-preserving measures of financial risk exposures solve the
challenge of measuring aggregate risk among multiple financial
institutions without encroaching on the privacy of any individual
institution.  Previous approaches to addressing this challenge require
trusted third parties, i.e., regulators, to collect, archive, and properly
assess systemic risk.  Apart from the burden this places on government
oversight, such an approach is also highly inefficient, requiring properly
targeted and perfectly timed regulatory intervention among an increasingly
complex and dynamic financial system.  Privacy-preserving measures can
promote more efficient ``crowdsourced'' responses to emerging threats to
systemic stability, enabling both regulators and market participants to
accurately monitor systemic risks in a timely and coordinated fashion,
creating a more responsive negative-feedback loop for stabilizing the
financial system.  This feature may be especially valuable for promoting
international coordination among multiple regulatory jurisdictions.  While
a certain degree of regulatory competition is unavoidable given the
competitive nature of sovereign governments, privacy-preserving measures
do eliminate a significant political obstacle to regulatory collaboration
across national boundaries.

Privacy-preserving risk measures have several other financial and
non-financial applications.  Investors such as endowments, foundations,
pension and sovereign wealth funds can use these measures to ensure that
their investments in various proprietary vehicles---hedge funds, private
equity, and other private partnerships---are sufficiently diversified and
not overly concentrated in a small number of risk factors.  Financial
auditors charged with the task of valuing illiquid assets at a given
financial institution can use these measures to compare and contrast their
valuations with the industry average and the dispersion of valuations
across multiple institutions.  Real-time indexes of the aggregate amount
of hedging activity in systemically important markets like the S\&P 500
futures contract may be constructed, which could have served as an early
warning signal for the ``Flash Crash'' of May 6, 2010.

More broadly, privacy-preserving measures of risk exposures may be useful
in other industries in which aggregate risks are created by individual
institutions and where maintaining privacy in computing such risks is
important for promoting transparency and innovation, such as healthcare,
epidemiology, and agribusiness.

\newpage

\bibliography{crypto}

\begin{thebibliography}{10}
\bibitem{protocol_note}
\bibinfo{note}{The Secure Sum Protocol discussed previously assume input are in
  range of [0,1]. The protocol works equally with any range and in this
  application we first turned the raw data in billions of dollars and then
  assumed a range of [0,10000] for all inputs. This implies that all arithmetic
  is done using modulo 30000 (3x10000).}
\end{thebibliography}
\bibliographystyle{naturemag}

\clearpage

\section*{Appendix}
In this appendix, we provide formal theorems and proofs of the security
guarantees ensured by the Secure-Sum and three Secure-Inner-Product
protocols, assuming semi-honest parties (possibly curious but following the protocol correctly). 
Extensions to malicious parties can be considered but are not discussed here. 

Secure-Inner-Product protocols 1 and 2 use a third dummy party to help
with the computations while Secure-Inner-Product protocol 3 does not. The
dummy party does not possess inputs or receives meaningful information but
simply helps with the computation (note that for the applications in mind,
the use of a dummy party does not represent a significant obstacle).
Secure-Inner-Product protocols 1 and 3 are defined on quantized data,
while Secure-Inner-Product protocol 2 applies directly to real-valued
data. Finally, Secure-Inner-Product protocol 1 provides
information-theoretic security, Secure-Inner-Product protocol 2 provides
`almost' information-theoretic security (as defined in Theorem
\ref{ip-thm2}) and both protocols require only elementary operations at a
computational level, while Secure-Inner-Product protocol 3 provides
cryptographic security (i.e., it relies on computational-hardness
assumptions) and uses OT protocols (hence non-elementary operations such
as RSA \cite{rsa} encryptions and decryptions).

An important benchmark for the practical consideration of secure protocols
is the number of communication rounds, which require exchange of data over
communications media such as the internet.  With a standard internet
connection and for arbitrary distances this can take no longer than 2--3
seconds but may also dominate the protocol running time. All protocols
proposed here require few communication rounds. The following table
summarizes these properties, where $n$ denotes the vector dimension and
$q$ the quantization level.

\begin{center}
  \begin{tabular}{ l ||  c | c | c | c | c | }
     Protocols & Security & Dummy party & Data & Rounds & Complexity \\ \hline \hline
    Secure-Sum  & IT & no & real & 2 & elem.\ op.\\ \hline \hline
    Secure-Inner-Product 1 & IT & yes & quantized & 3 & elem.\ op.\\ \hline
    Secure-Inner-Product 2 & almost IT & yes & real & 3 & elem.\ op. \\ \hline
     Secure-Inner-Product 3 & crypto & no & quantized & 3 & $O(nq^2)$ OT \\
    \hline
  \end{tabular}
\end{center}

\subsection*{Sum Protocols and Theorems}

For convenience, we restate the Secure-Sum protocol.

\vspace{2ex}\noindent%
{\bf Secure-Sum Protocol.}\\
Inputs: for $i=1,\dots,m$, party $i$ possesses the secret number $x_i \in [0,1]$.\\
Output: each party obtains $s=\sum_{i=1}^m x_i$ (where the addition is over the reals).\\
Protocol:
\begin{enumerate}
\item Each pair of parties exchange privately random numbers. Namely, for all $i,j$ with $i\neq j$, party $i$ provides to party $j$ a random number $R_{ij}$ drawn uniformly at random in $[0,m]$.
\item For each $i$, party $i$ adds to its secret number the random numbers it has received from other parties and subtract the random numbers it has provided to other parties.
In formula, party $i$ computes $S_i=x_i+\sum_{j\in \{1,\dots,m\} \atop{j \neq i}} R_{ji}  - \sum_{j\in \{1,\dots,m\} \atop{j \neq i}} R_{ij} \,\,\mathrm{mod}\, m$. Each party publicly reveals $S_i$.
\item Each party computes $S=\sum_{i=1}^m S_i \,\,\mathrm{mod}\,m$, which equals $s=\sum_{i=1}^m x_i$.
\end{enumerate}

One can define other variants and extensions of this protocol, in which fewer random numbers are exchanged to minimize information flow, or in which more information is exchanged to check the correctness of parties computations (one may also use virtual parties for that). 
\begin{thm}\label{sum-thm}
Let $x_1,\dots,x_m$ be $m$ privately owned real numbers. Let $i \in
\{1,\dots,m\}$ and $\mathrm{View}_i$ denote the view of party i obtained
from the Secure-Sum protocol with inputs $x_1,\dots,x_m$. The protocol
outputs the sum $s=\sum_{i=1}^m x_i$ and the distribution of
$\mathrm{View}_i$ depends on $x_1,\dots,x_m$ only through $s$ and $x_i$.
\end{thm}

We provide first the proof argument for $m=3$. Assume that party 1 collects all the data it possesses and received from other parties to try to learn something about their secret numbers. That is, party 1 possesses its secret number $x_1$, the numbers $R_{12},R_{13},R_{21},R_{31}$ exchanged in step 1, the numbers $S_1,S_2,S_3$ revealed in step 2 and the output sum $s$ (whose information is already contained in the $S_i$'s). From these, party 1 can subtract in $S_2,S_3$ the terms depending on $R_{12},R_{13},R_{21},R_{31}$ and obtain the right-hand side of
\begin{align}
x_2 + (R_{32}-R_{23})=S_2+(R_{21}-R_{12})  \,\,\mathrm{mod}\, 3 \label{eq1} \\
x_3- (R_{23}-R_{32})=S_3+(R_{31}-R_{13}) \,\,\mathrm{mod}\, 3 \label{eq2}
\end{align}
and this is all the information party 1 can gather about other parties secret numbers. Adding these equations provides $x_2+x_3=s-x_1$, i.e., what can be deduced from knowing the sum of the secret numbers. To see that nothing else can be inferred from \eqref{eq1} or \eqref{eq2}, note that $R_{32}-R_{23}$ is uniform on $[0,m]$. However, for any fixed number $x \in [0,1]$, if one adds to it a random number $R$ uniformly drawn in $[0,m]$, the number $x+R$ is also uniformly drawn in $[0,m]$. Therefore, \eqref{eq1} (or \eqref{eq2}) does not provide any further information about $x_2$ (or $x_3$).

\begin{proof}[Proof of Theorem \ref{sum-thm}]
All the arithmetic in this proof is modulo $m$.
We first check that the protocol computes indeed the sum. We set $R_{ii}=0$ for all $i$, to simply notations.
This is straightforward since $S_i=x_i + \sum_{j} ( R_{ji}  -  R_{ij} )$ and hence, $\sum_{i=1}^m S_i = \sum_{i=1}^m x_i$.
Let $\mathrm{View}_1$ be the protocol view of party 1, i.e.,
$$\mathrm{View}_1 =\{x_1,R_{1i}, R_{i1}, S_i, \, \forall 1 \leq i \leq m\}.$$
Party 1 can subtract the $R_{ij}$'s it has access to in the $S_i$'s, obtaining $\mathrm{View}'_1$ as a sufficient statistic for $\mathrm{View}_1$, where
$$\mathrm{View}'_1 =\{x_1, I_i,  \forall i \neq 1\}$$
and
\begin{align*}
I_i&=x_i +Z_i\\
Z_i&= \sum_{j \neq 1,i} (R_{ji}-R_{ij})
\end{align*}
Let us define $Z=[Z_2,\dots,Z_m]^t$ and $W= [R_{2},\dots,R_m]^t$, where $R_i$ contains all the $R_{ji}$ for which $j \neq i$ (in increasing order). Note that $Z$ and $W$ are a random vectors of dimension respectively $(m-1)\times 1$ and $m(m-1)\times 1$. We then have that
$$Z=A W- A \Pi W,$$
where $A$ is the $(m-1) \times m(m-1)$ matrix whose $i$-th row is filled with 0's except at columns $[i(m-2)+1,(i+1)(m-2)]$ where it is 1, and $\Pi$ is a permutation matrix. Note that the rank of $A$ and the rank of $M:=A(I-\Pi)$ is $m-2$, implying that $\mathrm{Im}(M)=\Sigma_2^m$, where
$$\Sigma_2^m :=\{ u_2,\dots,u_m \in [0,m]: \sum_{i=2}^m u_i = 0 \}.$$
Therefore, for any $z,d \in \Sigma_{m-1}$, there exists $w$ such that $Mw=d$ and
$$\pp\{ M W \leq z + d\} = \pp\{ M (W-w) \leq z \}= \pp\{ M W \leq z \}$$
where the second equality uses the fact that $W$ and $W-w$ are both i.i.d.\ uniform over $[0,m]$. This shows that $Z=MW$ is uniform over $\Sigma_2^m$ and $I=[I_2,\dots,I_m]$ is uniform over
$$\Sigma_2^m(x_2,\dots,x_m) :=\{ u_2,\dots,u_m \in [0,m]: \sum_{i=2}^m u_i = \sum_{i=2}^m x_i \}.$$ Therefore, the distribution of $\mathrm{View}'_1$, and hence of $\mathrm{View}_1$, depends only on $\sum_{i=2}^m x_i = s-x_1$ and $x_1$. By symmetry, the analogue conclusion holds for any parties, which concludes the proof of the theorem.
\end{proof}

\subsection*{Inner-Product Protocols and Theorems}
We now present secure protocols to compute the sample correlation, or
equivalently the inner product, between two real vectors. Recall that the
sample correlation of two vectors $x=\{x_i\}_{i=1}^t$ and
$y=\{y_i\}_{i=1}^t$ is given by
$$\rho(x,y)=\frac{ \sum_{i=1}^t x_i y_i - t \bar{x}\bar{y} }{(t-1) s_x s_y}= \sum_{i=1}^t \wtil{x}_i \wtil{y}_i $$
where $\bar{x}=\frac{1}{t} \sum_{i=1}^t x_i$, $s_x=(\frac{1}{t-1} \sum_{i=1}^t (x_i -  \bar{x})^2)^{1/2}$,
$\bar{y}=\frac{1}{t} \sum_{i=1}^t y_i$, $s_y=(\frac{1}{t-1} \sum_{i=1}^t (y_i -  \bar{y})^2)^{1/2}$,
$\wtil{x}_i=  \frac{1}{(t-1)^{1/2}} (x_i - \bar{x})/s_x$ and $\wtil{y}_i= \frac{1}{(t-1)^{1/2}} (y_i - \bar{y})/s_y$. 


\begin{definition}
We denote by $\mZ_q$ the set $\{0,1,\dots,q-1\}$, and by $\F_q$ the same set equipped with the Galois field operations when $q$ is a power of a prime. We define by $\Sigma_k(x,\F_q)$ the sets of $k$-tuples in $\F_q$ which add up to $x$, i.e., $$\Sigma_k(x,\F_q):=\{(y_1,\dots,y_k) \in \F_q^k: y_1 + \dots + y_k \mo q=x \}.$$
We may call the $y_i$'s to be shares of $x$.
\end{definition}

\vspace{2ex}\noindent%
{\bf Secure-Inner-Product Protocol 1.}\\
Common inputs: $q \in \mZ_+$ (the quantization level), $n \in \mZ_+$ (the vector dimensions) and $p$ a prime larger than $q^2n$.\\
Party 1 inputs: $x_1,\dots, x_n \in \mZ_q$.\\
Party 2 inputs: $y_1,\dots, y_n \in \mZ_q$.\\
Party 3 inputs: none.

\begin{enumerate}
\item For $i=1,\dots,n$,
party 1 splits $x_i$ in three shares $x_i(1)$, $x_i(2)$ and $x_i(3)$ uniformly drawn in $\Sigma_3(x_i,\F_p):=\{(a,b,c) \in \F_p^3: a+b+c \mo p=x_i \}$ and party 2 splits $y_i$ in three shares $y_i(1)$, $y_i(2)$ and $y_i(3)$ uniformly drawn in $\Sigma_3(y_i,\F_p)$. Party $1$ provides privately to party $2$ the shares $x_i(1),x_i(2)$ and privately to party $3$ the share $x_i(3)$. Party $2$ provides privately to party $1$ the shares $y_i(1),y_i(2)$ and privately to party $3$ the share $y_i(3)$.
\item  Party 1 sets $p_i(1)=(x_i(1)+x_i(3))(y_i(1)+y_i(2))  \mo p$ and $\rho(1)=\sum_{i=1}^n p_i(1) \mo p$, party 2 sets $p_i(2)=y_i(3)(x_i(1)+x_i(2))+ x_i(2)(y_i(1)+y_i(2)) \mo p$ and $\rho(2)=\sum_{i=1}^n p_i(2) \mo p$, and party 3 sets $p_i(3)=x_i(3)y_i(3) \mo p$ and $\rho(3)=\sum_{i=1}^n p_i(3) \mo p$.
For $m=1,2,3$, party $m$ splits $\rho(m)$ in three shares $\rho(m,1), \rho(m,2)$ and $\rho(m,3)$ uniformly drawn in $\Sigma_3(\rho(m),\F_p)$ and reveals privately $\rho(m,k)$ to party $k$, for $k=1,2,3$.

\item For $k=1,2,3$, party $k$ computes $R(k)= \sum_{m=1}^3 \rho(m,k) \mo p$. Parties 1 and 2 exchange $R(1)$ and $R(2)$ and party 3 provides $R(3)$ to parties 1 and 2. Parties 1 and 2 compute $R(1)+R(2)+R(3)=\sum_{i=1}^n x_i y _i$.
\end{enumerate}

\begin{thm}\label{ip-thm}
Let $x=[x_1,\dots,x_n]$ and $y=[y_1,\dots,y_n]$ be two privately owned
vectors on $\F_q^n$. Let $\mathrm{View}_1$ denote the view of party 1
obtained from the Secure-Inner-Product protocol 1 with inputs $x,y$. The
protocol outputs the inner product $\rho=\sum_{i=1}^n x_i y_i$ and the
distribution of $\mathrm{View}_1$ depends on $x,y$ only through $\rho$ and
$x$. The reciprocal result holds for party 2.
\end{thm}
\begin{proof}[Proof of Theorem \ref{ip-thm}]
The arithmetic is on $\F_p$ in the following. We first check that the
protocol computes indeed the inner product. For every $i=1,\dots,n$,
$p_i(1)+p_i(2)+p_i(3)=x_i y_i$, hence
$$\sum_{i=1}^n (p_i(1)+p_i(2)+p_i(3))=\rho(1)+\rho(2)+\rho(3)=\sum_{i=1}^n x_i y_i.$$
Moreover, $\sum_{k=1}^3 \rho(m,k)=\rho(m)$, hence
$$\sum_{k=1} R(k)=\sum_{k=1} \sum_{m=1}^3 \rho(m,k) = \sum_{m=1}^3 \rho(m)=\sum_{i=1}^n x_i y_i.$$
Let $\mathrm{View}_1$ be the protocol view of party 1, which is a function of
$$\mathrm{View}_1' =\{x, y(1), y(2), \rho(2,1), \rho(3,1), R(2), R(3) \},$$
where $y(1)$ contains all components $y_i(1)$ for $i=1,\dots,n$ and similarly for the $y(2)$. Note that for $i=1,\dots,n$, $(p_i(1), p_i(2), p_i(3))$ are independent and uniformly drawn in $\Sigma_3(p_i,\F_p)$, where $p_i=x_iy_i$. Moreover, step 2.\ and 3.\ of the protocol are equivalent to running the secure-sum-protocol on $\rho(1),\rho(2),\rho(3)$. Hence, from Theorem \ref{sum-thm}, for any realization of $\rho(1),\rho(2),\rho(3)$, the distribution of $\rho(2,1), \rho(3,1), R(2), R(3)$ depends only on the sum $\rho(1)+\rho(2)+\rho(3)=\sum_{i=1} p_i$ and on $\rho(1)$, where $\rho(1)$ depends only on $x$ and on $y(1), y(2)$ which are independent and uniformly distributed over $\F_p$. Therefore, the distribution of $\mathrm{View}_1'$, hence $\mathrm{View}_1$, depends only on $\sum_{i=1} p_i=\rho$ and on $x$.
\end{proof}

\vspace{2ex}\noindent%
{\bf Secure-Inner-Product Protocol 2.}\\
Common input: $n \in \mZ_+$ (the vector dimensions) and $\tau \geq n$\\
Party 1 inputs: $x_1,\dots, x_n \in [0,1]$.\\
Party 2 inputs: $y_1,\dots, y_n \in [0,1]$.\\
Party 3 inputs: none.

\begin{enumerate}
\item For $i=1,\dots,n$,
\begin{enumerate}
\item party 1 splits $x_i$ in three shares by evaluating a random polynomial $t \mapsto X_i(t)$ at $(t_1,t_2,t_3)=(1/4,1/2,3/4)$, where $X_i(t)=x_i + a_i t \mo \tau$ and where $a_i$ is uniformly drawn in $[0,\tau]$.  Party 1 reveals $X_i(t_j)$ to party $j$ for $j=2,3$,
\item party 2 splits $y_i$ in three shares $Y_i(t_j)=y_i+b_i t_j \mo \tau$, for $j=1,2,3,$ where
$b_i$ is uniformly drawn in $[0,\tau]$, and reveals $Y_i(t_j)$ to party $j$ for $j=1,3$.
\end{enumerate}
\item For $j=1,2,3$,
\begin{enumerate}
\item party $j$ computes $P(t_j)=\sum_{i=1}^t X_i(t_j)Y_i(t_j) \mo \tau$,
\item party $j$ draws $\alpha_j,\beta_j$ independently and uniformly at random in $[0,\tau]$ and for $k=1,2,3$, sets $Z_j(t_k)=\alpha_j t_k + \beta_j t_k^2 \mo \tau$ and shares $Z_j(t_k)$ with party $k$,
\item $\rho(t_j)=P(t_j) + \sum_{k=1}^3 Z_k(t_j) \mo \tau$ is made available to parties 1 and 2.
\end{enumerate}
\item Party 1 and 2 compute $\rho(0)$ by interpolating a degree 2 polynomial on $\rho(t_j)$, $j=1,2,3$, obtaining $\rho(0)=\sum_{i=1}^n x_i y_i$.
\end{enumerate}

\begin{thm}\label{ip-thm2}
Let $x=[x_1,\dots,x_n]$ and $y=[y_1,\dots,y_n]$ be two privately owned
real vectors on $[0,1]^n$, where $n$ is fixed. Let $\mathrm{View}_1$
denote the view of party 1 obtained from the Secure-Inner-Product protocol
2 (over the reals) with inputs $x,y$. The protocol outputs the inner
product $\rho=\sum_{i=1}^n x_i y_i$ and the distribution of
$\mathrm{View}_1$ can be approximated  arbitrarily close (in total
variation distance and when $\tau$ increases) by a distribution depending
on $x,y$ only through $\rho$ and $x$. The reciprocal result holds for
party 2.
\end{thm}

\noindent%
We omit the proof of this theorem to conserve space since it does not
concern the main scope of the paper. We refer to Theorem \ref{ip-thm} for
a proof of a Secure Inner-Product protocol, which can be used on real data
via quantization.

We provide a third protocol to compute securely the inner-product function
without using a third dummy party but ensuring only cryptographic
security. This protocol uses the Oblivious Transfer (OT) protocol,
developed by \cite{rabinOT,OT}, which is an important protocol for
multi-party computations as it allows to compute in particular secret
shares of the product $x\cdot y$ of two bits $x$ and $y$, and can then be
used in the computation of more general circuit computations. The basic OT
protocol allows a sender to transfer one of potentially many bits to a
receiver; however, the sender remains oblivious as to what bit the
receiver wants and the receiver remains oblivious about any other bits
than the one he has requested. In other words, the functionality in the OT
protocol takes the bits $(b_1,\dots,b_k)$ as inputs for the first party
and the index $i$ for the second party, and produces as output nothing for
the first party and the bit $b_i$ requested by the second party. Formally,
$$\text{OT}_1^k((b_1,\dots,b_k),i)=(\lambda,b_i),$$
where $\lambda$ denotes the no information symbol.
We now describe OT$_1^2$.

\subsubsection*{OT$_1^2$ protocol}

Sender inputs: $(b_0,b_1) \in \{0,1\}^2$ and a private key $(n,d)$.\\
Receiver inputs: $i \in \{0,1\}$ and a public key $(n,e)$.\\

\noindent
Algorithm:
\begin{enumerate}
\item The sender generates two random numbers $x_0,x_1$ and transmit them to the receiver.
\item The receiver generates a random number $k$, encrypts it with the public key and scrambles the outcome with $x_i$ to produce $c=(x_i+k^e) \mod n$
\item The sender decrypts the two numbers $(c-x_0)$ and $(c-x_1)$ to get $k_0$ and $k_1$ respectively (i.e., it computes $k_j=(c-x_j)^d \mod n$ for $j=0,1$). Note that either $k_0$ or $k_1$ is equal to $k$, but these are equally likely for the sender, and reciprocally, $k_{i\oplus 1}$ is not accessible to the receiver. The sender then transmits $a_0=b_0 + k_0$ and $a_1=b_1 + k_1$.
\item The receiver finds $b_i=a_i - k$.
\end{enumerate}
The OT$_1^k$ protocol is easily obtained by extending previous protocol to multiple sender bits, ad similarly, one can extend the protocol to non binary fields.

We now present a cryptographic protocol for the inner product.

\vspace{2ex}\noindent%
{\bf Secure-Inner-Product Protocol 3.}\\
Common inputs: $q$ (the quantization level), $n$ (the vector dimensions).\\
Party 1 inputs: $x_1,\dots, x_n \in \mZ_q$.\\
Party 2 inputs: $y_1,\dots, y_n \in \mZ_q$.

\begin{enumerate}
\item For $i=1,\dots,n$,
\begin{enumerate}
\item party 1 picks $x_i(2)$ uniformly at random in $\mZ_{nq^2}$ and reveals it to party 2,
who picks $y_i(1)$ uniformly at random in $\mZ_{nq^2}$ and reveals it to party 1.
\item party 1 picks $a_{i}(1)$ uniformly at random in $\mZ_{nq^2}$ and sends
$$\{-a_{i}(1), -a_{i}(1) + x_i(1), -a_{i}(1) + 2 x_i(1), -a_{i}(1) + 3x_i(1),\dots, -a_{i}(1) + (nq^2 -1) x_i(1)\}$$
(all operations $\mo nq^2$) with OT$_1^{nq^2}$ to party 2 who picks the $y_i(2)$-th element.
\item party 2 picks $b_{i}(2)$ uniformly at random in $\mZ_{nq^2}$ and sends
$$\{-b_{i}(2), -b_{i}(2) + x_i(2), -b_{i}(2) + 2 x_i(2), -b_{i}(2) + 3x_i(2),\dots, -b_{i}(2) + (tq^2 -1) x_i(2)\}$$
(all operations $\mo nq^2$) with OT$_1^{nq^2}$ to party 1 who picks the $y_i(1)$-th element.
\item party 1 computes $p_i(1)=x_i(1)y_i(1)+a_i(1)+b_i(1) \mo nq^2$ and party computes $p_i(2)=x_i(2)y_i(2)+a_i(2)+b_i(2) \mo nq^2$. Note that these are shares of the product $x_iy_i$.
\end{enumerate}
\item Party 1 computes $\rho(1)=\sum_{i=1}^n p_i(1) \mo nq^2$ and reveals it to party 2, who computes $\rho(2)=\sum_{i=1}^n p_i(2) \mo nq^2$ and reveals it to party 1.
\item Each party computes $\rho(1)+\rho(2) \mo nq^2=\sum_{i=1}^n x_i y_i$.
\end{enumerate}

\begin{figure}
\begin{center}
\includegraphics[scale=.4]{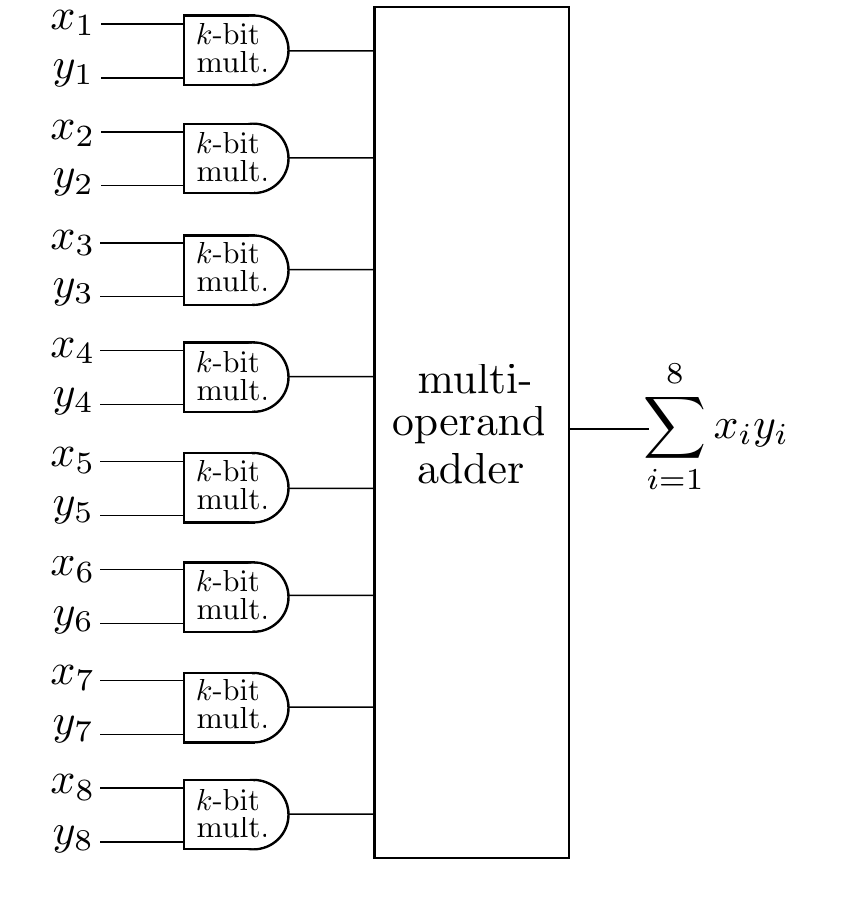}
\caption{Computational circuit for the inner product $\sum_{i=1}^8 x_i y_i$, when the inputs are $k$-bit numbers.}
\label{circ_operand}
\end{center}
\end{figure}

\noindent%
From the protocol construction, we have the following result.
\begin{lemma}
Secure-Inner-Product protocol 3 privately reduces the correlation
computation to the OT protocol.
\end{lemma}
The notion of being ``privately reducible'' is formally defined in Section
2.2.\ of \cite{goldreich}. From the composition theorem for the
semi-honest setting in Section 2.2.\ of \cite{goldreich}, one obtains as a
consequence of the previous lemma that Secure-Inner-Product protocol 3
privately computes the inner product provided the existence of trapdoor
one-way permutations. In particular, using RSA for the encryptions in OT,
the protocol is secure provided that RSA cannot be broken.

This protocol requires $O(nq^2)$ OT protocols but only three communication
rounds. This still means a possibly high number of public and private
encryptions/decryptions (e.g., with RSA). One may use \cite{amort} to
improve the OT protocols running time. Another approach consist in using a
Boolean circuit for correlations as in Figure \ref{circ_operand}, using OT
protocols to compute shares of the multiplication gates (and simply adding
shares for the XOR gates). Such an approach, as developed in \cite{gmw},
or related approaches as in \cite{yao, bmr}, may be particularly useful
for other functions such as for the quantile function, which does not have
the arithmetic structure of the summation or inner-product functions. In
particular, \cite{yao,bmr} provide protocols with constant communication
rounds which may matter for practical considerations, although for real
data problems, the practicality of such algorithms need to be further
investigated.

\subsection*{Related literature on MPCs}\label{sec_literature}
\subsubsection*{Theory}
The problem of secure multi-party computation emerged with the work of Yao
\cite{yao} in 1982, and with the work of Goldreich, Micali and Wigderson
\cite{gmw} in 1987. It is shown in \cite{yao} that any Boolean
functionality can be computed without requiring an external trusted party
for two parties, and \cite{gmw} provides protocols for arbitrarily many
parties. Since these papers, many have proposed variations of MPC
settings, allowing different kinds of adversarial parties, security, and
efficiency attributes. In particular, \cite{bmr} introduces cryptographic
protocols with bounded circuit depths (requiring finitely many
communication rounds) and \cite{bgw,ccd,cddhr} develop
information-theoretic protocols. Homomorphic encryption has also been
shown to provide another approach to secure multi-party computations
\cite{cramer,franklin}, and more recently, Gentry \cite{gentry} showed
that fully homomorphic encryption schemes can be constructed, allowing
addition and multiplication to be performed on encrypted data without
having to decrypt it. This approach leads to MPC protocols that do not
have communication rounds increasing with the circuit complexity, although
fully homomorphic encryption is still considered impractical. For certain
functionality, progress regarding practical fully homomorphic encryption
have been achieved in \cite{bv} with somewhat fully homomorphic
encryptions schemes using the learning-with-errors assumption.

\subsubsection*{Applications}
The main applications associated with MPCs in the literature include
distributed voting \cite{vote}, private bidding and auctions
\cite{auction}, data mining \cite{datamining}, and sharing of signature
\cite{sign}. MPCs have been used for the first time in a real-world
application only in 2008, when 1,200 farmers in Denmark employed an MPC
protocol in a nation-wide auction to determine the market price of
sugar-beets contracts without revealing their selling and buying prices
\cite{live}. The whole computation took about half an hour, a satisfactory
time for this application. In a different context, \cite{pce} introduces
``Patient Controlled Encryption'' scheme, where an electronic health record
system allowing searches to be done on encrypted data is developed.

\end{document}